\documentclass[journal]{IEEEtran}
\IEEEoverridecommandlockouts

\usepackage{amssymb}
\usepackage{bbm}
\usepackage{amsmath}
\usepackage{caption}
\usepackage{cite}
\usepackage{mathrsfs}
\usepackage[displaymath,mathlines]{lineno}
\usepackage{color}
\usepackage{overpic}
\usepackage{algorithm,algpseudocode}
\usepackage{algorithmicx}
\usepackage{pbox}
\usepackage{multicol}
\usepackage{amsthm}
\usepackage{epstopdf}
\usepackage[most]{tcolorbox}

\usepackage{lipsum}

\makeatletter
\newcommand{\doublewidetilde}[1]{{%
		\mathpalette\double@widetilde{#1}%
}}
\newcommand{\double@widetilde}[2]{%
	\sbox\z@{$\m@th#1\widetilde{#2}$}%
	\ht\z@=.5\ht\z@
	\widetilde{\box\z@}%
}
\makeatother

\newtheorem{definition}{Definition}

\newtheorem{lemma}{Lemma}
\newtheorem{corollary}{Corollary}

\begin{document}
\title{\LARGE Reconfigurable Intelligent Surface-Assisted Massive MIMO: Favorable Propagation, Channel Hardening, and Rank Deficiency}
\author{Trinh~Van~Chien, Hien~Quoc~Ngo, Symeon~Chatzinotas,  and Bj\"orn~Ottersten
}
\maketitle
%\IEEEpeerreviewmaketitle
\vspace{-16mm}
%\begin{abstract}
\par Massive multiple-input multiple-output (MIMO) and reconfigurable intelligent surface (RIS) are two promising technologies for 5G-and-beyond wireless networks, capable of providing large array gain and multiuser spatial multiplexing. Without requiring additional frequency bands, those technologies offer significant improvements in both spectral and energy efficiency by simultaneously serving many users. The performance analysis of an RIS-assisted Massive MIMO system as a function of the channel statistics relies heavily on fundamental properties including favorable propagation, channel hardening, and rank deficiency. The coexistence of both direct and indirect links results in aggregated channels, whose properties are the main concerns of this lecture note. For practical systems with a finite number of antennas and engineered scattering elements of the RIS, we evaluate the corresponding deterministic metrics with  Rayleigh fading channels as a typical example.
%\end{abstract}

%\vspace{-1mm}
\section*{Relevance}
%\vspace{-1mm}
\par Antenna arrays and propagation environments are  key factors fundamentally determining the performance of wireless communication systems. Massive MIMO communications has demonstrated the possibilities of increasing the communication throughput by coherently processing many antenna signals at each base station (BS) compared to the number of served users \cite{Marzetta2016a}. Each antenna element of a massive array can contribute an extra degree of freedom to spatial processing that allows the system to obtain array gain and suppress mutual interference \cite{Rusek2013a}. Massive MIMO, thus, offers unprecedented improvements even under hardware impairments by exploiting the different channels to different users \cite{massivemimobook}. Nonetheless, many scenarios limit the channel capacity despite deploying massive antenna arrays, especially under harsh conditions such as large blockages, which result in ill-conditioned channels \cite{Chien2021TWC}. 

To improve the performance of Massive MIMO under harsh propagation conditions, the use of reconfigurable intelligent surface (RIS) implemented through meta surfaces is very promising. RIS is an emerging technology that smartly controls the propagation  environments  by  a  planar  array  with  many engineered scattering  elements  to  form electromagnetic  waves in  the desired  structure \cite{wu2019intelligent,9326394}. Each engineered scattering  element  induces  a  phase  shift  on the incident signals and reflects them passively without the requirement of radio frequency chains. Thus, with RIS, high power consumption and expensive hardware can be avoided. In the scenarios of large-distance/heavy blockages between the transmitter and the receiver, RIS has been shown to enhance the received signal strength thanks to a phase-shift design that leads to a constructive combination of multiple arriving waves at the receiver, and therefore yields better system performance than the absence of an RIS \cite{huang2019reconfigurable}. 

Since RIS-assisted Massive MIMO is a very new topic, there is no standard reference presenting and providing the effective properties when many antennas and phase-shift elements are installed in the system. This lecture note fills this gap by presenting the channel property aspects of the RIS-assisted Massive MIMO systems when the numbers of BS antennas and engineered scattering elements grow large. Three important properties, namely favorable propagation, channel hardening, and rank deficiency, will be exploited. Achieving such understanding will directly enable us to design a robust and energy-efficient RIS-assisted Massive MIMO system.
%\vspace*{-1mm}
\section*{Prerequisites}
%\vspace*{-1mm}
\par Basic knowledge of random variables, linear algebra, signals and systems, RIS, and Massive MIMO are required.

\begin{figure}[t]
	\centering
	\includegraphics[trim=7cm 8.5cm 10.5cm 3.0cm, clip=true, width=2.8in]{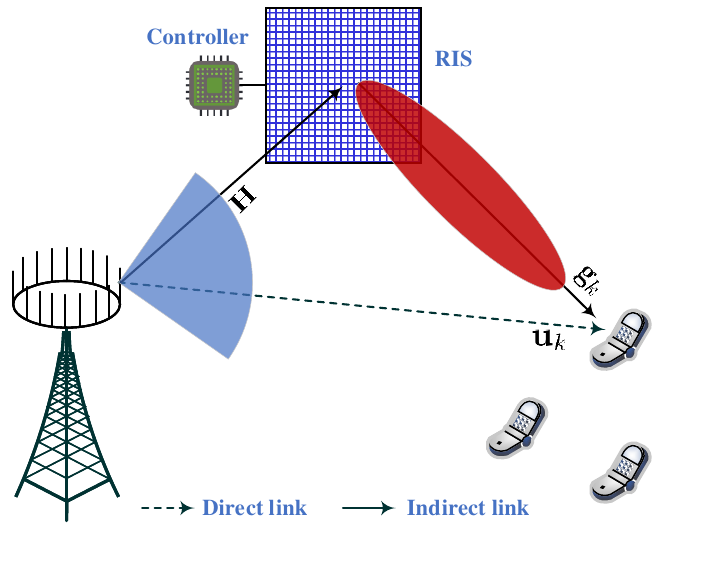} \vspace*{-0.1cm}
	\caption{An RIS-assisted Massive MIMO system.}
	\label{FigSysModelIRS}
	\vspace*{-0.2cm}
\end{figure}
%\vspace*{-0.1cm}
\section*{Problem Statement}
%\vspace*{-0.1cm}
\par We consider an RIS-aided communication system where a BS is equipped with $M$ antennas to serve $K$ single-antenna users as shown in Fig.~\ref{FigSysModelIRS}. To enhance the system performance, an RIS with $N$ engineered scattering elements is deployed in the coverage area between the BS and the users. The phase-shift matrix is denoted as $\mathbf{\Phi} = \mathrm{diag}(e^{j\theta_1}, \ldots, e^{j \theta_N})$ with $-\pi \leq \theta_n \leq \pi, \forall n,$ representing the phase shift of the $n$-th passive engineered scattering element that is controllable. For the indirect link from the BS to a user through the RIS (see Fig.~\ref{FigSysModelIRS}), let us denote $\mathbf{H} \in \mathbb{C}^{M \times N}$ the channel matrix between the BS and the RIS, while $\mathbf{g}_k \in \mathbb{C}^{N}$ denotes the channel vector between the RIS and user~$k$. For the direct link between the BS and user~$k$ (see Fig.~\ref{FigSysModelIRS}), the propagation channel is $\mathbf{u}_k \in \mathbb{C}^M$. We now consider the uplink transmission where all users simultaneously transmit signals to the BS. Let $s_k$ with $\mathbb{E}\{ |s_k|^2 \} = 1$ be the symbol transmitted by user k. Then, the received signal at the BS, denoted by $\mathbf{y} \in \mathbb{C}^M,$ is\footnote{The same methodology should be straightforwardly extended to the downlink transmission, especially when the time division duplexing (TDD) protocol is exploited.}
\begin{equation} \label{eq:ReceivedSignal}
\begin{split}
\mathbf{y} &=  \sqrt{p}\sum\limits_{k=1}^K  \mathbf{u}_k s_k +  \sqrt{p} \mathbf{H} \pmb{\Phi} \sum\limits_{k=1}^K  \mathbf{g}_k s_k  + \mathbf{w} \\
&=  \sqrt{p}  \sum\limits_{k=1}^K \mathbf{z}_k s_k + \mathbf{w},
\end{split}
\end{equation}
where $p$ is the transmit power allocated to each data symbol, $\mathbf{w} \sim \mathcal{CN}(\mathbf{0}, \sigma^2 \mathbf{I}_M)$ is additive noise, and
\begin{equation} \label{eq:AggregatedChannel}
\mathbf{z}_k = \mathbf{u}_k + \mathbf{H} \mathbf{\Phi} \mathbf{g}_k,
\end{equation}
which is the aggregated channel of user~$k$ comprising both the direct and indirect inks. Channels need to be acquired for signal processing. One conventional method is that the BS-RIS channels and the RIS-user channels are estimated separately. This is very challenging because the RIS lacks a digital processing unit. Another method for channel acquisition is that the BS only needs to estimate the aggregated channels $\mathbf{z}_k$. This method yields a great benefit of reducing the system cost with the same pilot training overhead as communication systems without the presence of the RIS. One example with fine details step-by-step on the aggregated channel estimation can be found in \cite{Chien2021TWC}. Therefore, the aggregated channel statistics are of particular interest for performance evaluation purposes.  The problem at hand is to investigate the channel properties when the numbers of BS antennas and engineered scattering elements grow large, i.e., $M, N \rightarrow \infty,$ with the same speed.
%\end{tcolorbox}
%\vspace*{-0.2cm}
\section*{Solution}

%\begin{tcolorbox}
The aggregated channels do not follow the same distributions as the conventional Massive MIMO channels in general due to the weighted product from the indirect link in \eqref{eq:AggregatedChannel}. Note that, with the absence of the RIS, the aggregated channels reduce to the conventional Massive MIMO channels, i.e., $\mathbf{z}_k = \mathbf{u}_k, \forall k$. Consequently, the results presented in this lecture note are a general version of the conventional Massive MIMO communications thanks to the presence of the RIS. In the following, the favorable propagation, channel hardening, and rank deficiency properties will be discussed.
\subsection*{Favorable Propagation}
%\vspace*{-0.2cm}
\par This section presents the favorable propagation property in RIS-assisted Massive MIMO systems. Specifically, it describes the orthogonality among the aggregated channel vectors from the BS to the $K$ users as shown in the following definition.  
\begin{definition} \label{DefOptFP}
The aggregated channels offer the favorable propagation property if the inner product of two arbitrary channel vectors $\mathbf{z}_k$ and $\mathbf{z}_{l}$, $k\neq l,$ satisfies
\begin{equation}
 \mathbf{z}_k^H \mathbf{z}_{l} = 0.
\end{equation}
\end{definition}
Favorable propagation is very important. If the channels are favorable, then the signals transmitted from the $K$ users will belong to $K$ orthogonal spaces. As a result, the BS can decode the signal sent by user~$k$ without inter-user interference by simply projecting  the received signal \eqref{eq:ReceivedSignal} on to $\mathbf{z}_k^H$ as
\begin{equation} \label{eq:DetectSig}
\begin{split}
\hat{s}_k &= \mathbf{z}_k^H \mathbf{y} = \sqrt{p} \mathbf{z}_k^H \mathbf{z}_k s_k + \sqrt{p} \sum\limits_{l=1, l \neq k}^K  \mathbf{z}_{k}^H \mathbf{z}_{l}s_{l} + \mathbf{z}_k^H \mathbf{w} \\
& = \sqrt{p} \| \mathbf{z}_k \|^2 s_k + \mathbf{z}_k^H \mathbf{w},
\end{split}
\end{equation}
The above projection corresponds to the linear maximum-ratio combing (MRC) technique. This implies the favorable propagation property yields optimal signal detection performance with only a simple linear processing. Nonetheless, this property will rarely hold in practice. The \textit{asymptotically favorable propagation} given in Definition~\ref{Def:FP2} is, therefore, more interested.
\begin{definition} \label{Def:FP2}
As $M, N \rightarrow \infty$, the aggregated channels offer the asymptotically favorable propagation property if
\begin{equation} \label{eq:FPV1}
 \frac{\mathbf{z}_k^H\mathbf{z}_{l}}{\sqrt{\mathbb{E}\{ \| \mathbf{z}_k \|^2 \}  \mathbb{E}\{ \| \mathbf{z}_{l} \|^2 \} }} \rightarrow 0,
\end{equation}
holds for any $k \neq l$.
\end{definition}
In contrast to the strict requirement in Definition~\ref{DefOptFP}, which the inner product $\mathbf{z}_k^H \mathbf{z}_{l}$ must be exactly zero, the condition~\eqref{eq:FPV1} gives practical tolerability in the sense that $\mathbf{z}_k^H \mathbf{z}_{l}$ does not need to be zero for finite values. For real systems, the BS and the RIS are equipped with a large (but finite) number of BS antennas/scattering elements. Conditioned on an aggregated channel model that offers the asymptotically favorable propagation, we now evaluate how this property behaves with a finite value $M$ and $N$ by the deterministic favorable channel propagation. 
\begin{definition}
%\vspace*{-0.2cm}
For the given channel vectors $\mathbf{z}_k$ and $\mathbf{z}_{l}$, the deterministic favorable propagation metric is defined as
\begin{equation} \label{eq:DeFP}
\mathsf{FP}_{kk'}= \frac{ \mathsf{Var} \{ \mathbf{z}_k^H \mathbf{z}_{l} \}}{ \mathbb{E}\{ \| \mathbf{z}_k \|^2 \} \mathbb{E}\{ \| \mathbf{z}_{l} \|^2 \}}.
\end{equation}
\end{definition}
By evaluating \eqref{eq:DeFP}, we enable to measure how close to the favorable propagation the channel is when the number of antennas is finite. In other words, \eqref{eq:DeFP} represents the speed of convergence in \eqref{eq:FPV1}.
It is worth to emphasize that this deterministic metric is independent of small-scale fading coefficients and therefore by knowing the large-scale fading coefficients and phase shifts, we can evaluate the favorable propagation property of the aggregated channels.
%\begin{tcolorbox}

\textbf{The first lesson learned}: The favorable propagation means the aggregated channel vectors from the BS to different users are pair-wisely orthogonal. Under the favorable propagation, the optimal performance can be achieved with simple linear processing techniques. The deterministic metric \eqref{eq:DeFP} measures how close to the favorable propagation the aggregated channels offer with a finite number of BS antennas and engineered scattering elements.
%\end{tcolorbox}
%\vspace*{-0.1cm}
\subsection*{Channel Hardening} 
%\vspace*{-0.1cm}
This section represents the channel hardening property which is a phenomenon where the norm square of the aggregated channel vectors from the BS to the users do not fluctuate much (even the small-scale fading channels randomly change). Mathematically, an aggregated channel offers channel hardening if the condition in Definition~\ref{DefChannelHardening} holds.
\begin{definition} \label{DefChannelHardening}
As $M, N \rightarrow \infty$, an aggregated channel model offers the channel hardening property if 
\begin{equation} \label{eq:ChannelHardening}
\frac{\| \mathbf{z}_k\|^2}{\mathbb{E}\{ \| \mathbf{z}_k \|^2 \}} \rightarrow 1,
\end{equation}
with the almost sure convergence.
\end{definition}
From \eqref{eq:ChannelHardening}, the effective channel gain $\| \mathbf{z}_k \|^2$ can be replaced by its mean value by virtue of many antennas at the BS and many scattering elements at the RIS. If the channels harden, in the uplink, the BS can replace the instantaneous channel gain by its mean value for signal detection. This significantly simplifies the signal processing as well as resource allocation designs at the BS because all designs now can be done on the large-scale fading time scale. With TDD, the importance of channel hardening is even more significant in the downlink data transmission. In the downlink, thanks to the channel hardening property, each user can treat the mean value of the effective channel gain as the true one to detect the desired signal. Thus, no downlink pilot overhead is required for the downlink channel estimation.

Since $\mathbf{z}_k$ is a random vector, it is nontrivial to exploit Definition~\ref{DefChannelHardening} to prove the channel hardening property of the channels. In order to seek for a more tractable metric, we employ the Chebyshev inequality \cite{ngo2017no} such that\footnote{If $x$ is a random variable with the mean $\bar{x}$ and variance $\sigma_{x}^2$, the Chebyshev inequality gives $\mathsf{Pr}(| x - \bar{x} | \geq \epsilon) \leq \sigma_x^2/\epsilon^2$.}
\begin{equation} \label{eq:Chebysev}
\begin{split}
&\mathsf{Pr} \left\{ \left| \frac{\| \mathbf{z}_k\|^2}{\mathbb{E} \{ \| \mathbf{z}_k \|^2\}} - 1 \right|^2 \leq \epsilon \right\} \\
&= 1 - \mathsf{Pr} \left\{ \left| \frac{\| \mathbf{z}_k\|^2}{\mathbb{E} \{ \| \mathbf{z}_k \|^2\}} - 1 \right|^2 \geq \epsilon \right\} \\
&\geq 1 - \frac{\mathsf{Var}\{ \| \mathbf{z}_k \|^2 \}}{\epsilon^2  (\mathbb{E} \{ \| \mathbf{z}_k \|^2 \} )^2 },
\end{split}
\end{equation}
for any $\epsilon > 0$. Observing the last inequality of \eqref{eq:Chebysev}, one can have the channel hardening if
\begin{equation} \label{eq:VarMeanRatio}
\frac{\mathsf{Var}\{ \| \mathbf{z}_k \|^2 \}}{ (\mathbb{E} \{ \| \mathbf{z}_k \|^2 \} )^2 } \rightarrow 0.
\end{equation}
Otherwise, the channel model does not harden. Consequently, \eqref{eq:VarMeanRatio} can be utilized to determine if a channel model offers the channel hardening property as shown in Definition~\ref{DefChannelHardeningMetric}.
\begin{definition} \label{DefChannelHardeningMetric}
For each aggregated channel, the deterministic channel hardening metric is defined as
\begin{equation} \label{eq:CHk}
\mathsf{CH}_k = \frac{ \mathsf{Var}\{ \| \mathbf{z}_k \|^2 \}}{ (\mathbb{E} \{ \| \mathbf{z}_k \|^2 \} )^2}.
\end{equation}
\end{definition}
The deterministic channel hardening metric in \eqref{eq:CHk} can be computed for practical propagation channel models. In some particular scenarios we can even derive this metric in the closed form expression. 

\textbf{The second lesson learned}: The channel hardening property makes each effective channel gain approach its mean value, with a high probability. The deterministic metric \eqref{eq:CHk} measures the hardening level of the channel vectors with a finite number of BS antennas and engineered scattering elements.

%\vspace*{-0.2cm}
\subsection*{Rank Deficiency}
%\vspace*{-0.2cm}
This section presents the rank deficiency, which has been a fundamental issue in many wireless communication aspects. A rich scattering propagation environment may not often offer a channel profile with the rank deficiency. Nonetheless, this may appear in a poor scattering environment such as mmWave communications with a limited number of reflections and propagation paths. For RIS-assisted Massive MIMO systems, we provide the rank deficiency definition of an aggregated channel as in Definition~\ref{Def:ChannelSparsity}.
\begin{definition} \label{Def:ChannelSparsity}
The rank deficiency indicates fewer channel degrees of freedom than the upper limit imposed by the number of BS antennas and engineered scattering elements. Mathematically, we can measure the number of channel degrees of freedom for an aggregated channel as 
 $\mathrm{rank}(\widehat{\mathbf{R}}_k)$,
where $\widehat{\mathbf{R}}_k= \mathbb{E} \{ \mathbf{z}_k \mathbf{z}_k^H \} \in \mathbb{C}^{M \times M}$ is the covariance matrix of the aggregated channel of user~$k$. The rank deficiency implies $\mathrm{rank}(\widehat{\mathbf{R}}_k) < M$.
\end{definition}
The rank of a covariance matrix equals to the number of non-zero eigenvalues, so projecting the covariance matrices utilizing eigendecomposition and analyzing the eigenvalues is one way to observe the rank deficiency. Each eigenvalue stands for the average squared amplitude in a subset of angular directions, the rank deficiency results in zero eigenvalues or almost zero eigenvalues (compressible eigenvalues). Consequently, we can  project a given covariance matrix onto the eigendomain with the eigendecomposition as
$\widehat{\mathbf{R}}_k = \mathbf{V}_k \pmb{\Sigma}_k \mathbf{V}_k^H$,
where $\mathbf{V}_k \in \mathbb{C}^{M \times M}$ is an unitary matrix including the eigenvectors of $\widehat{\mathbf{R}}_k$ and the diagonal matrix $\pmb{\Sigma}_k = \mathrm{diag}([\delta_1, \ldots, \delta_M]^T) \in \mathbb{C}^{M \times M}$ includes the eigenvalues, $\delta_1, \ldots, \delta_M,$ on the diagonal. The rank deficiency appears if a covariance matrix has at least one zero eigenvalue that reduces the system capacity.

We now demonstrate the rank-improved benefits of installing an RIS in the vicinity of a Massive MIMO system by first observing from \eqref{eq:DetectSig} that the effective power from user~$k$ scales up with $\|\mathbf{z}_k \|^2$. One can use the channel hardening property to approximate the effective gain $\|\mathbf{z}_k \|^2$ by the mean value $\mathbb{E} \{ \|  \mathbf{z}_k \|^2 \}$. We further process the mean value as
\begin{equation} \label{eq:zkpower}
\begin{split}
&\mathbb{E} \{ \|\mathbf{z}_k \|^2 \} =  \mathbb{E} \{ \| \mathbf{u}_k \|^2 \} + \mathbb{E} \big\{ \mathbf{u}_k^H  \mathbf{H} \pmb{\Phi} \mathbf{g}_k \big\}  \\
&+ \mathbb{E} \big\{ \mathbf{g}_k^H \pmb{\Phi}^H \mathbf{H}^H  \mathbf{u}_k \big\} + \mathbb{E} \big\{ \big| \mathbf{H}\pmb{\Phi} \mathbf{g}_k \big|^2 \big\} \\
& \stackrel{(a)}{\geq} \mathbb{E} \{ \| \mathbf{u}_k \|^2 \} + \mathbb{E} \big\{ \big| \mathbf{H}\pmb{\Phi} \mathbf{g}_k \big|^2 \big\} \\
&\stackrel{(b)}{=} \mathrm{tr}\big( \mathbb{E} \{ \mathbf{u}_k \mathbf{u}_k^H \} \big) + \mathrm{tr}\big( \mathbb{E} \big\{ \mathbf{H} \pmb{\Phi} \mathbf{g}_k \mathbf{g}_k^H \pmb{\Phi}^H \mathbf{H}^H \big\} \big),
\end{split}
\end{equation}
where $(a)$ is by designing the phase shift matrix $\pmb{\Phi}$ to obtain $\mathbb{E} \big\{ \mathbf{u}_k^H  \mathbf{H} \pmb{\Phi} \mathbf{g}_k \big\} + \mathbb{E} \big\{ \mathbf{g}_k^H \pmb{\Phi}^H \mathbf{H}^H  \mathbf{u}_k \big\} \geq 0$; $(b)$ is obtained by the trace of a product property. The observation from \eqref{eq:zkpower} is that an RIS can improve the received signal power on average since 
$\mathbb{E} \{ \|\mathbf{z}_k \|^2 \} \geq \mathbb{E} \{ \| \mathbf{u}_k \|^2 \}$,
which the equality holds when RIS is not active or without contributions. Since the covariance matrices are positive semi-definite, one possibly designs the phase shifts to obtain
\begin{equation}
\mathrm{rank} \big( \mathbb{E} \{ \mathbf{z}_k \mathbf{z}_k^H \} \big) \geq \mathrm{rank} \big( \mathbb{E} \{ \mathbf{u}_k \mathbf{u}_k^H \} \big),
\end{equation}
which the equality holds, for example, if the covariance matrix is full-rank that happens in rich scattering environments. In contrast, RIS is beneficial to a propagation environment suffering from the rank deficiency issue.  

%\begin{tcolorbox}
\textbf{The third lesson learned}: The rank deficiency can be explicitly observed by analyzing the eigenvalues of each covariance matrix. Moreover, RIS can be potentially utilized to compensate for the rank deficiency appearing in Massive MIMO communications.  
%\end{tcolorbox}
%\vspace*{-0.1cm}

\subsection*{Example: Rayleigh Fading Channels}
%\vspace*{-0.1cm}
This section considers a spatially correlated Rayleigh model, which is aligned with a BS radiating isotropically in RIS. All the channels between the BS and user~$k$ are defined as
\begin{equation} \label{eq:ChannelModel1}
	\mathbf{H} = \mathbf{R}_s^{1/2} \widetilde{\mathbf{H}} \mathbf{R}_{si}^{1/2}, \, \, \mathbf{g}_k =  \mathbf{R}_{ik}^{1/2} \tilde{\mathbf{g}}_k, \mbox{ and } \mathbf{u}_k = \mathbf{R}_k^{1/2} \tilde{\mathbf{u}}_k,
\end{equation}
where $\widetilde{\mathbf{H}} \in \mathbb{C}^{M \times N}, \tilde{\mathbf{g}}_k \in \mathbb{C}^{N},$ and $\tilde{\mathbf{u}}_k \in \mathbb{C}^{M}$ represent the small-scale fading matrix/vectors whose elements are independent and identically distributed $\mathcal{CN}(0,1)$; $\mathbf{R_s} \in \mathbb{C}^{M \times M}, \mathbf{R}_{si} \in \mathbb{C}^{N \times N}, \mathbf{R}_{ik} \in \mathbb{C}^{N \times N}$, and $\mathbf{R}_k \in \mathbb{C}^{M \times M}$ represent the  covariance matrices. With a limited scatterer number at the BS, we can model the corresponding covariance matrices by, for instance, the approximate Gaussian local scattering model \cite{massivemimobook} such that the $(m,n)-$th element is given as
\begin{equation} \label{eq:Cov1}
[\mathbf{R}_{\nu}]_{mn} = \frac{\beta_\nu}{S_\nu}\sum\limits_{s=1}^{S_\nu} e^{j \pi (m-n) \sin (\psi_s) } e^{- \frac{\sigma_\nu^2}{2} \big( \pi (m-n)  \cos(\psi_s) \big)^2 },
\end{equation}
where $\nu \in \{ s, k \}$ and $\beta_\nu$ is the corresponding large-scale fading coefficient; $S_\nu$ is the number of scattering clusters around the BS and $\psi_s$ is the nominal angle of arrival (AoA) for cluster~$s$. The multipath components of a cluster is assumed to have AoAs following a Gaussian distribution that is distributed around the nominal AoA with the angular standard deviation $\sigma_\nu$. By assuming that the RIS is fabricated in a rectangular surface of size $N_H \times N_V$ where $N_H$ and $N_V$ are the numbers of elements in each column and row, respectively, the covariance matrices $\mathbf{R}_{si}$ and $\mathbf{R}_{ik}$ are given as
\begin{equation} \label{eq:Cov2}
\mathbf{R}_{si} = \mathbf{R},  \mathbf{R}_{ik} = \beta_{ik} \mathbf{R},
\end{equation}
where $\beta_{ik}$ is the large-scale fading coefficient associated with the cascaded channel $\mathbf{g}_k$. The spatial correlation matrix $\mathbf{R} \in \mathbb{C}^{N \times N}$ has the $(m,n)-$th element defined as 
%\begin{equation}
	$[\mathbf{R}]_{mn} = d_Hd_V\mathrm{sinc} (2 \|\mathbf{v}_m - \mathbf{v}_n \|/ \lambda)$,
%\end{equation}
where $\mathbf{v}_x = [0, \mod(x-1,N_H)d_H, \lfloor (x-1)/N_H \rfloor d_V ]^T$ with $x \in \{m,n\}$ \cite{bjornson2020rayleigh,van2021outage}; $\lambda$ is the wavelength of a plane wave; $\mod(\cdot)$ and $\lfloor \cdot \rfloor$ are the modulus operation and floor function, respectively; each phase-shift element is of size $d_H \times d_V$ where $d_V$ is the vertical height and $d_H$ is the horizontal width. 

As a special case, the spatially uncorrelated fading channel model gives $\mathbf{R}_\nu = \beta_\nu \mathbf{I}_M$ and $\mathbf{R} =  d_H d_V \mathbf{I}_N$, whose instantaneous channels are defined as
\begin{equation} \label{eq:UncorrRayleigh}
\mathbf{H} =  \sqrt{\beta_sd_H d_V} \widetilde{\mathbf{H}}, \mathbf{g}_k =  \sqrt{\beta_{ik} d_H d_V} \tilde{\mathbf{g}}_k, \mathbf{u}_k =  \sqrt{\beta_k} \tilde{\mathbf{u}}_k,
\end{equation}
which ignores the spatial correlation among elements that is, therefore, less practical. Even though the original channels follow circularly symmetric Gaussian distributions, each aggregated channel follows a non-Gaussian due to the weighted product of multiple Gaussian random variables. For a given phase-shift matrix, we still can compute the channel statistics as in Lemma~\ref{Lemma:Momments}.
\begin{lemma} \label{Lemma:Momments}
	If the channel model \eqref{eq:ChannelModel1} is utilized, the second and forth moments of each aggregated channel vector in \eqref{eq:AggregatedChannel} are respectively given as
	\begin{align}
		&\mathbb{E}\{ \| \mathbf{z}_k \|^2 \} = M \beta_k +  M \beta_s \mathrm{tr}(\widetilde{\pmb{\Theta}}_k), \label{eq:zk2}\\
		&\mathbb{E}\{ \| \mathbf{z}_k \|^4 \} = \big( M \beta_k  +  M \beta_s \mathrm{tr}(\widetilde{\pmb{\Theta}}_k) \big)^2 +   2 \mathrm{tr}(\mathbf{R}_s \mathbf{R}_k) \mathrm{tr}(\widetilde{\pmb{\Theta}}_k) \notag \\
		&+  \big|\mathrm{tr}(\widetilde{\pmb{\Theta}}_k ) \big|^2\mathrm{tr}\big(\mathbf{R}_s^2\big) + \mathrm{tr} \big(\widetilde{\pmb{\Theta}}_k^2 \big) \big( M^2 \beta_s^2 + \mathrm{tr}\big(\mathbf{R}_s^2\big)\big) + 
		\mathrm{tr}(\mathbf{R}_k^2),\label{eq:zk4v1}
	\end{align}
	where $\widetilde{\pmb{\Theta}}_k= \pmb{\Phi}^H \mathbf{R}_{si} \pmb{\Phi} \mathbf{R}_{ik}$.
\end{lemma}
\begin{proof}
	The second moment of the aggregated channel of user~$k$ is first processed by the independence of the direct and indirect links as
	\begin{equation}
		\begin{split}
			\mathbb{E}\{ \| \mathbf{z}_k \|^2 \} &= \mathbb{E} \{ \| \mathbf{u}_k \|^2 \} + \mathbb{E} \left\{ \| \mathbf{H} \pmb{\Phi} \mathbf{g}_k \|^2 \right\} \\
			& =  \mathrm{tr}( \mathbf{R}_k) +   \mathbb{E} \left\{ \mathbf{g}_k^H \pmb{\Phi}^H \mathbf{H}^H \mathbf{H} \pmb{\Phi} \mathbf{g}_k \ \right\} \\
			&= M \beta_k +  M \beta_s \mathrm{tr}(\widetilde{\pmb{\Theta}}_k),
		\end{split}
	\end{equation}
	where the second moment of the indirect link is computed by using the channel model \eqref{eq:ChannelModel1} together with \cite[Lemma~8]{Chien2020book}. Next, the forth moment of the aggregated channel of user~$k$ is computed as
	\begin{equation} \label{eq:zk4}
		\begin{split}
			&\mathbb{E}\{ \| \mathbf{z}_k \|^4 \} = \mathbb{E} \{ \| \mathbf{u}_k +  \mathbf{H} \pmb{\Phi} \mathbf{g}_k \|^4 \}  =  \mathbb{E} \{ | a + b + c + d  |^2 \} \\
			&  = \mathbb{E}\{ |a|^2 \} + \mathbb{E}\{ |b|^2 \} + \mathbb{E}\{ |c|^2 \} + 2 \mathbb{E}\{ ad \} + \mathbb{E}\{|d|^2\},
		\end{split}
	\end{equation}
	where $a = \| \mathbf{u}_k \|^2,$ $b = \mathbf{u}_k^H \mathbf{H} \pmb{\Phi} \mathbf{g}_k,$ $c= \mathbf{g}_k^H \pmb{\Phi}^H \mathbf{H}^H \mathbf{u}_k,$ and $d = \| \mathbf{H} \pmb{\Phi} \mathbf{g}_k \|^2$. The first expectation in the last equality of \eqref{eq:zk4} is computed by deploying \cite[Lemma~9]{Chien2020book} as
	\begin{equation} \label{eq:uk4v1}
		\mathbb{E} \{ \| \mathbf{u}_k \|^4 \} = | \mathrm{tr} ( \mathbf{R}_k ) |^2 + \mathrm{tr}( \mathbf{R}_k^2) = M^2 \beta_k^2 + \mathrm{tr}( \mathbf{R}_k^2).
	\end{equation}
	Meanwhile, one uses \cite[Lemma~8]{Chien2020book} to compute the next three expectations in the last equality of \eqref{eq:zk4} as
	\begin{equation} \label{eq:bcad}
		\begin{split}
		&\mathbb{E}\{ |b|^2 \} = \mathbb{E} \{ |c|^2 \} = \mathrm{tr}(\mathbf{R}_s \mathbf{R}_k) \mathrm{tr}(\widetilde{\pmb{\Theta}}_k),\\
		&\mathbb{E} \{ ad \} =  \mathrm{tr}( \mathbf{R}_s) \mathrm{tr}( \mathbf{R}_k)  \mathrm{tr}(\widetilde{\pmb{\Theta}}_k) = M^2 \beta_s \beta_k \mathrm{tr}(\widetilde{\pmb{\Theta}}_k).
		\end{split}
	\end{equation}
	We tackle the last expectation in the last equality of \eqref{eq:zk4} as
	\begin{equation} \label{eq:d2}
		\begin{split}
			&\mathbb{E}\{ |d|^2 \} = \mathbb{E} \left\{ \| \mathbf{H} \pmb{\Phi} \mathbf{g}_k \|^4 \right\} \\
			&=  \mathbb{E} \left\{ \left| \mathbf{g}_k^H \pmb{\Phi}^H \mathbf{R}_{ik}^{1/2} \widetilde{\mathbf{H}}^H \mathbf{R}_s \widetilde{\mathbf{H}} \mathbf{R}_{ik}^{1/2} \pmb{\Phi} \mathbf{g}_k \right|^2 \right\} \\
			&=  \mathbb{E} \left\{ \left\|\mathbf{R}_{ik}^{1/2} \pmb{\Phi} \mathbf{g}_k\right\|^4 \left| \mathbf{t}_k^H \mathbf{R}_s  \mathbf{t}_k \right|^2 \right\},
		\end{split}
	\end{equation}
	where $\mathbf{t}_k= \widetilde{\mathbf{H}} \mathbf{R}_{ik}^{1/2} \pmb{\Phi} \mathbf{g}_k/ \|\mathbf{R}_{ik}^{1/2} \pmb{\Phi} \mathbf{g}_k \|$. Since  $\mathbf{t}_k \sim \mathcal{CN}(\mathbf{0}, \mathbf{I}_M)$ and noting that $\mathbf{t}_k$ and $\mathbf{g}_k$ are independent circularly symmetric Gaussian vectors, we recast \eqref{eq:d2} to as
	\begin{equation} \label{eq:d2v1}
		\begin{split}
			&\mathbb{E}\{ |d|^2 \} = \ \mathbb{E} \left\{ \left\|\mathbf{R}_{ik}^{1/2} \pmb{\Phi} \mathbf{g}_k\right\|^4 \right\} \mathbb{E}\big\{ \big| \mathbf{t}_k^H \mathbf{R}_s \mathbf{t}_k \big|^2 \big\} \\
			&= \Big( \big|\mathrm{tr}\big(\widetilde{\pmb{\Theta}}_k \big) \big|^2 + \mathrm{tr} \big(  \widetilde{\pmb{\Theta}}_k^2 \big) \Big)\left( |\mathrm{tr}(\mathbf{R}_s)|^2 + \mathrm{tr}(\mathbf{R}_s^2)\right) \\
			&= \left( \big|\mathrm{tr}\big(\widetilde{\pmb{\Theta}}_k \big) \big|^2 + \mathrm{tr} \big(  \widetilde{\pmb{\Theta}}_k^2 \big) \right)\big( M^2 \beta_s^2 + \mathrm{tr}(\mathbf{R}_s^2)\big),
		\end{split}
	\end{equation}
	Plugging \eqref{eq:uk4v1}, \eqref{eq:bcad}, and \eqref{eq:d2v1} into \eqref{eq:zk4}, we obtain the result as shown in the lemma.
\end{proof}
Note that the even moments in Lemma~\ref{Lemma:Momments} are multivariate functions of phase shifts and covariance matrices. Consequently, this allows us to further investigate other utility metrics such as spectral efficiency and outage probability in closed form that is independent of small-scale fading coefficients.  By  means of  \cite[Theorem~II.1]{lasserre1995trace}, we observe the power scaling law for a single-user system (user~$k$ only) over the spatially correlated channels as   
\begin{multline} \label{eq:Bound}
M \beta_k + M \beta_s \sum_{n=1}^N \lambda_n (\pmb{\Phi}^H \mathbf{R}_{si} \pmb{\Phi}) \lambda_{N-n+1} (\mathbf{R}_{ik})   \leq	 \\
  \mathbb{E}\{ \| \mathbf{z}_k \|^2 \} \leq   M \beta_k + M \beta_s \underbrace{\sum_{n=1}^N \lambda_n (\pmb{\Phi}^H \mathbf{R}_{si}\pmb{\Phi}) \lambda_{n} (\mathbf{R}_{ik}  )}_{\triangleq \gamma_k },
\end{multline}
where $\{ \lambda_n (\pmb{\Phi}^H \mathbf{R}_{si}\pmb{\Phi} ) \}_{1}^{N}$ and $ \{ \lambda_{n} (\mathbf{R}_{ik}  ) \}_1^N$ are the eigenvalues of the Hermitian matrices $\pmb{\Phi}^H \mathbf{R}_{si}\pmb{\Phi}$ and $\mathbf{R}_{ik} $ sorted in descending order. We observe that
\begin{equation} \label{eq:Upperbound}
\begin{split}
\gamma_k & \stackrel{(a)}{\leq} \left( \sum_{n=1}^N \lambda_n (\pmb{\Phi}^H \mathbf{R}_{si}\pmb{\Phi}) \right) \left( \sum_{n=1}^N \lambda_{n} (\mathbf{R}_{ik}  ) \right) \\
&\stackrel{(b)}{=}  \mathrm{tr}(\pmb{\Phi}^H \mathbf{R}_{si}\pmb{\Phi}) \mathrm{tr}(\mathbf{R}_{ik}) \stackrel{(c)}{=} N^2 d_H^2 d_V^2 \beta_{ik},
\end{split}
\end{equation}
where $(a)$ is obtained because the eigenvalues are non-negative; $(b)$ is due to the trace and eigenvalues relationship, $(c)$ is acquired by the covariance matrix structure in \eqref{eq:Cov2} and since the phase matrix is unitary. From \eqref{eq:Bound} and \eqref{eq:Upperbound},  we observe that $\mathbb{E}\{ \| \mathbf{z}_k \|^2 \}$ scales up with an array gain provided by the RIS in the order from $N$ to $N^2$ depending upon the phase shift design. Specifically, the right-hand of \eqref{eq:Bound} is upper bounded by $M\beta_k + MN^2 d_H^2 d_V^2 \beta_s \beta_{ik} $, whose scaling law is possibly  obtained by the optimal phase shift design \cite{wu2019intelligent}. For the case of uncorrelated Rayleigh fading channels, the channel statistics are given in Corollary~\ref{CorUncorrMomentsv1}.
\begin{figure}[t]
	\begin{minipage}{0.5\textwidth}
		\centering
		\includegraphics[trim=0.4cm 0cm 0.0cm 0.5cm, clip=true, width=3.5in]{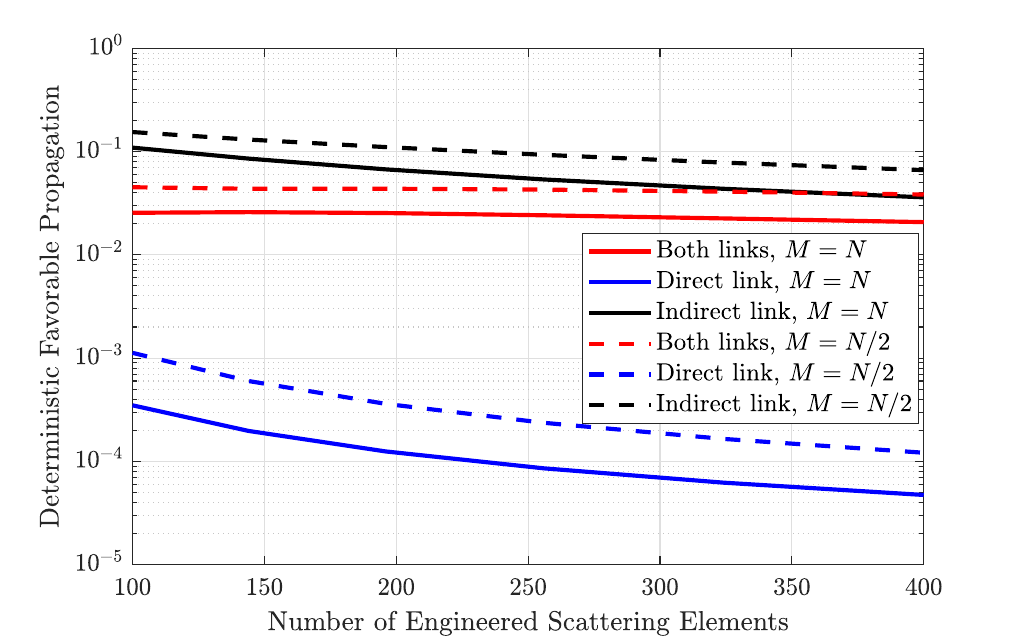} %\vspace*{-0.1cm}
		\\ $(a)$
		%\vspace*{-0.2cm}
	\end{minipage}
	\begin{minipage}{0.5\textwidth}
		\centering
		\includegraphics[trim=0.4cm 0cm 0.0cm 0.5cm, clip=true, width=3.5in]{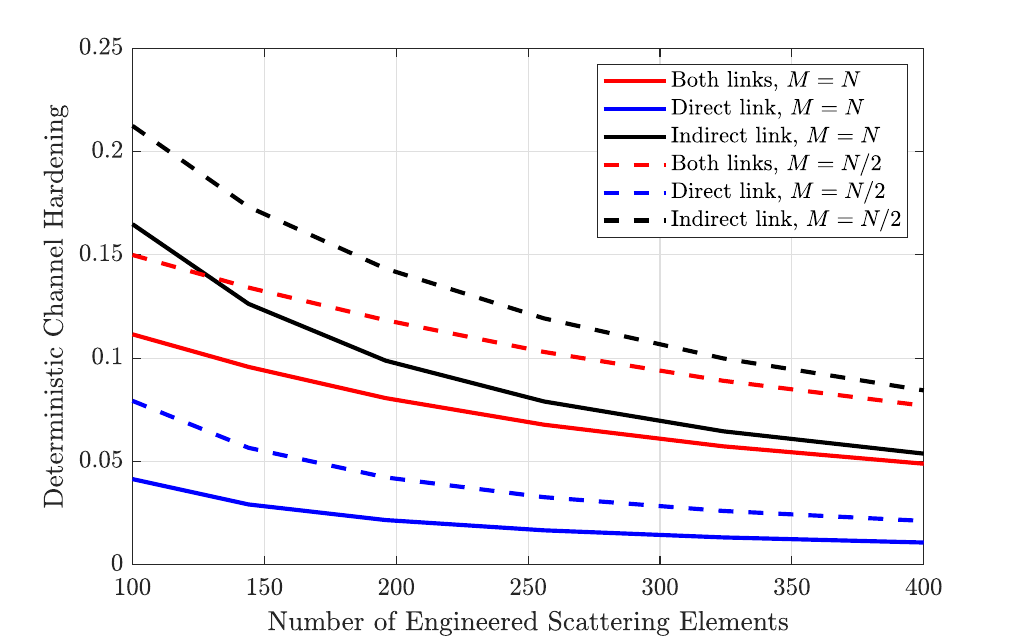} %\vspace*{-0.1cm}
		\\$(b)$
		%\vspace*{-0.2cm}
	\end{minipage}
	\begin{minipage}{0.5\textwidth}
		\centering
		\includegraphics[trim=0.4cm 0cm 0.0cm 0.5cm, clip=true, width=3.5in]{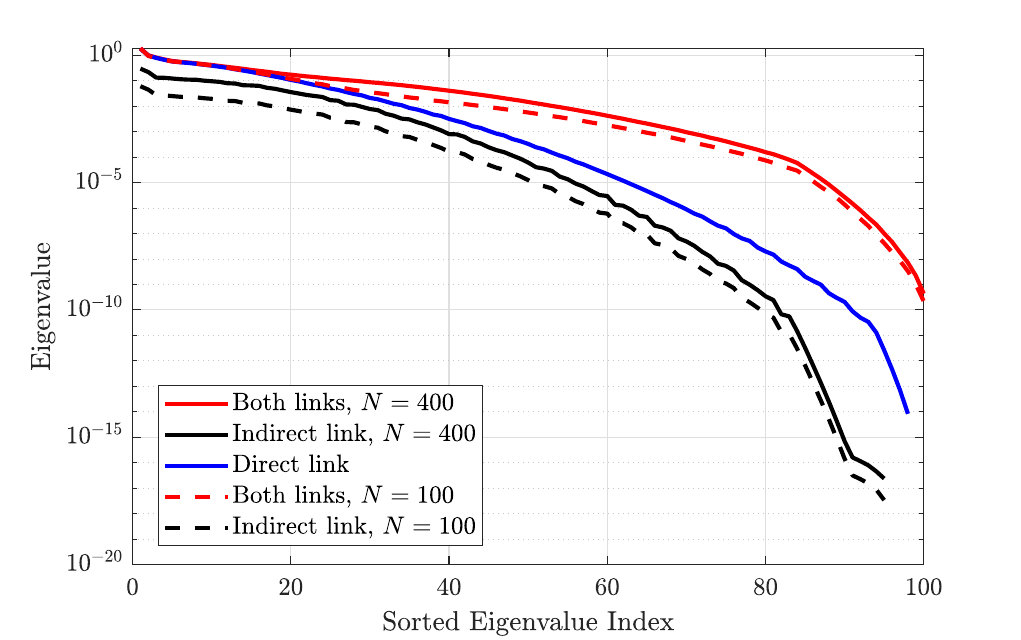}%\vspace*{-0.1cm}
		\\$(c)$
		%\vspace*{-0.1cm}
	\end{minipage}
	\caption{The fundamentals of an RIS-assisted Massive MIMO system with a finite number of BS antennas and engineered scattering elements of the RIS: $(a)$ The deterministic favorable propagation metric; $(b)$ The deterministic channel hardening metric; and $(c)$ The sorted eigvenvalues with $M=100$.}
	\label{Fig2}
	\vspace{-0.2cm}
\end{figure}
\begin{corollary}\label{CorUncorrMomentsv1}
If the channel model \eqref{eq:UncorrRayleigh} is utilized, the second and forth moments of each aggregated channel vector in \eqref{eq:AggregatedChannel} are respectively given as
\begin{align}
	&\mathbb{E}\{ \| \mathbf{z}_k \|^2 \} = M \beta_k + MN \beta_s \beta_{ik} d_H^2 d_V^2, \label{eq:zk2unc}\\
	&\mathbb{E}\{ \| \mathbf{z}_k \|^4 \} = M^2 \beta_k^2 + 4MN\beta_k \xi_k  \notag\\
	& +  (M^2 +M)(N^2 + N) \xi_k^2 + M \beta_k^2.
\end{align}
with $\xi_k = \beta_s \beta_{ik} d_H^2 d_V^2$.
\end{corollary}
\begin{proof}
By substituting the covariance matrices on the definition into \eqref{eq:zk2} and \eqref{eq:zk4v1} with $\mathrm{tr}(\widetilde{\pmb{\Theta}}) = N \beta_{ik} d_H^2 d_V^2 $ and $\mathrm{tr}(\widetilde{\pmb{\Theta}}^2) = N  \beta_{ik}^2 d_H^4 d_V^4$, we arrive at the desired results in Corollary~\ref{CorUncorrMomentsv1}.
\end{proof}
Lemma~\ref{lemmaFP} concludes the favorable propagation of the spatially correlated Rayleigh fading model. 
\begin{lemma} \label{lemmaFP}
	By assuming that $ \liminf_{N \rightarrow \infty} \mathrm{tr}( \pmb{\Phi}^H \mathbf{R} \pmb{\Phi} \mathbf{R} )/ N >0$ and $\limsup_{N \rightarrow \infty} \| \pmb{\Phi}^H \mathbf{R} \pmb{\Phi} \mathbf{R} \|_2 < \infty$, the spatially correlated Rayleigh channel model offers the asymptotically favorable propagation property and the deterministic favorable propagation metric can be computed in the closed form expression as
	\begin{equation} \label{eq:FPkkprimeClosed}
		\mathrm{FP}_{kl}= \frac{\mathrm{tr} \big( \big( \mathbf{R}_k +  \mathrm{tr}(\widetilde{\pmb{\Theta}}_k) \mathbf{R}_s \big)\big(  \mathbf{R}_{l} +  \mathrm{tr}(\widetilde{\pmb{\Theta}}_{l}) \mathbf{R}_s  \big) \big)}{\big(  M \beta_k +  M \beta_s \mathrm{tr}(\widetilde{\pmb{\Theta}}_k) \big)\big( M \beta_{l} +  M \beta_s \mathrm{tr}(\widetilde{\pmb{\Theta}}_{l}) \big)}.
	\end{equation} 
\end{lemma}
\begin{proof}
By utilizing the aggregated channel on the definition in \eqref{eq:AggregatedChannel} together with its statistic information in Corollary~\ref{CorUncorrMomentsv1}, we first observe \eqref{eq:Asymptkkprime}.
\begin{figure*}
\begin{equation} \label{eq:Asymptkkprime}
\frac{|\mathbf{z}_k^H\mathbf{z}_{l}| }{ \sqrt{\mathbb{E}\{ \| \mathbf{z}_k \|^2 \}  \mathbb{E}\{ \| \mathbf{z}_{l} \|^2 \} }} \leq \frac{1}{\alpha_{kl}} \left| \frac{ \mathbf{u}_k^H \mathbf{u}_{l} }{MN} + \frac{\mathbf{u}_k^H \mathbf{H} \pmb{\Phi} \mathbf{g}_{l} }{MN} + \frac{\mathbf{g}_k^H \pmb{\Phi}^H \mathbf{H}^H  \mathbf{u}_{l} }{MN} +  \frac{\mathbf{g}_k^H \pmb{\Phi}^H \mathbf{H}^H \mathbf{H} \pmb{\Phi}  \mathbf{g}_{l} }{MN} \right| \rightarrow 0, \mbox{ as }  M, N \rightarrow \infty 
\end{equation}
%\hrule
\vspace*{-0.7cm}
\end{figure*}
In \eqref{eq:Asymptkkprime}, $\alpha_{kl} = \beta_s d_H^2 d_V^2 \lambda \sqrt{\beta_{ik} \beta_{il} }$, where $\tilde{\lambda} = \liminf_{N \rightarrow \infty} \mathrm{tr}( \pmb{\Phi}^H \mathbf{R} \pmb{\Phi} \mathbf{R} )/ N$. The obtained result in \eqref{eq:Asymptkkprime} implies that the asymptotically favorable propagation in \eqref{eq:FPV1} holds.

For a random variable $X$, its variance is computed as $\mathsf{Var}\{ X\} = \mathbb{E} \{ |X|^2 \} - \big| \mathbb{E} \{ X\} \big|^2 $. Thus, $\mathsf{Var} \big\{ \mathbf{z}_k^H \mathbf{z}_{l} \big \} = \mathbb{E} \big\{ |\mathbf{z}_k^H \mathbf{z}_{l}|^2 \big\}$ due to the zero mean of the aggregated channels. Next, we have
	\begin{equation}\label{eq:Ezkzkprime}
		\mathbb{E} \big\{ |\mathbf{z}_k^H \mathbf{z}_{k'}|^2 \big\} = \mathrm{tr}\left( \mathbb{E}\big\{\mathbf{z}_{k'} \mathbf{z}_{k'}^H \big\} \mathbb{E}\big\{ \mathbf{z}_k \mathbf{z}_k^H \big\} \right).
	\end{equation}
	From \eqref{eq:Ezkzkprime}, we now compute $\mathbb{E} \{ \mathbf{z}_k \mathbf{z}_k^H\}$. We have,
	\begin{equation} \label{eq:zkzk}
		\begin{split}
			& \mathbb{E}\big\{ \mathbf{z}_k \mathbf{z}_k^H \big\} = \mathbb{E}\{ \mathbf{u}_k \mathbf{u}_k^H \} + \mathbb{E}\left\{ \mathbf{H} \pmb{\Phi} \mathbf{g}_k \mathbf{g}_k^H \pmb{\Phi}^H \mathbf{H}^H \right\} \\
			&= \mathbf{R}_k +  \mathbb{E}\left\{ \mathbf{R}_s^{1/2} \widetilde{\mathbf{H}}  \mathbf{R}_{si}^{1/2} \pmb{\Phi} \mathbf{R}_{ik} \pmb{\Phi}^H \mathbf{R}_{si}^{1/2} \widetilde{\mathbf{H}}^H \mathbf{R}_s^{1/2} \right\} \\
			&=  \mathbf{R}_k +  \mathrm{tr}(\widetilde{\pmb{\Theta}}_k) \mathbf{R}_s,
		\end{split}
	\end{equation}
	where the last equality in \eqref{eq:zkzk} is obtained by utilizing \cite[Lemma~8]{Chien2020book}. In a similar manner, we obtain $\mathbb{E}\{ \mathbf{z}_{l}  \mathbf{z}_{l}^H \} = \mathbf{R}_{l} +  \mathrm{tr}(\widetilde{\pmb{\Theta}}_{l}) \mathbf{R}_s$. Hence, we obtain \eqref{eq:Ezkzkprime} in the closed form as
	\begin{equation} \label{eq:zkzkprimeclosed}
		\mathbb{E}\big\{ \big| \mathbf{z}_k \mathbf{z}_{l}^H \big|^2 \big\} = \mathrm{tr} \big( \big( \mathbf{R}_k +  \mathrm{tr}(\widetilde{\pmb{\Theta}}_k) \mathbf{R}_s \big)\big( \mathbf{R}_{l} +  \mathrm{tr}(\widetilde{\pmb{\Theta}}_{l}) \mathbf{R}_s  \big) \big).
	\end{equation}
Substituting \eqref{eq:zk2} and \eqref{eq:zkzkprimeclosed} into \eqref{eq:DeFP}, we obtain the result as shown in the lemma and conclude the proof.
\end{proof}
We now exploit the identity $\mathrm{tr}(\mathbf{X}+\mathbf{Y}) = \mathrm{tr}(\mathbf{X}) +  \mathrm{tr}(\mathbf{Y})$ for the two matrices $\mathbf{X}$ and $\mathbf{Y}$, the numerator of \eqref{eq:FPkkprimeClosed}, denoted by $\mathsf{Num}_{kl}$, is reformulated as
\begin{equation} \label{eq:Numkkprime}
\begin{split}
\mathsf{Num}_{kl} =&  \mathrm{tr}(\mathbf{R}_k \mathbf{R}_{l} ) + \mathrm{tr}(\widetilde{\pmb{\Theta}}_k) \mathrm{tr}(\mathbf{R}_s \mathbf{R}_{l}) \\
&+   \mathrm{tr}(\widetilde{\pmb{\Theta}}_{l})\mathrm{tr}(\mathbf{R}_k \mathbf{R}_s) +  \mathrm{tr}(\widetilde{\pmb{\Theta}}_k) \mathrm{tr}(\widetilde{\pmb{\Theta}}_{l}) \mathrm{tr}(\mathbf{R}_s^2).
\end{split}
\end{equation}
Meanwhile, denoting by $\mathsf{Den}_{kl}$ the denominator of \eqref{eq:FPkkprimeClosed} and it is reformulated as
\begin{equation}
	\begin{split}
		\mathsf{Den}_{kl} =&   \mathrm{tr}(\mathbf{R}_k)  \mathrm{tr}(\mathbf{R}_{l}) + \mathrm{tr}(\widetilde{\pmb{\Theta}}_k) \mathrm{tr}(\mathbf{R}_s) \mathrm{tr}(\mathbf{R}_{l})  \\
		&+  \mathrm{tr}(\widetilde{\pmb{\Theta}}_{l})\mathrm{tr}(\mathbf{R}_k)\mathrm{tr}(\mathbf{R}_s)+ \mathrm{tr}(\widetilde{\pmb{\Theta}}_k)\mathrm{tr}(\widetilde{\pmb{\Theta}}_{l}) \big(\mathrm{tr}(\mathbf{R}_s) \big)^2.
	\end{split}
\end{equation}
Since the covariance and phase-shift matrices are all positive semi-definite, we can use the identity $\mathrm{tr}(\mathbf{X}\mathbf{Y}) \leq \mathrm{tr}(\mathbf{X})\mathrm{tr}(\mathbf{Y})$ for the positive semi-definite matrices to obtain $\mathsf{Num}_{kl} \leq \mathsf{Den}_{kl}$. Dividing the numerator and denominator of \eqref{eq:FPkkprimeClosed} by $N^2 M^2$, one obtains that
\begin{multline} \label{eq:Limv1}
\lim_{M,N \rightarrow \infty} \mathrm{FP}_{kl} = \\
 \lim_{M,N \rightarrow \infty} \frac{\mathrm{tr}(\widetilde{\pmb{\Theta}}_k) \mathrm{tr}(\widetilde{\pmb{\Theta}}_{l}) \mathrm{tr}(\mathbf{R}_s^2) /(N^2 M^2)}{\mathrm{tr}(\widetilde{\pmb{\Theta}}_k)\mathrm{tr}(\widetilde{\pmb{\Theta}}_{l}) \big(M \beta_s \big)^2/(N^2 M^2)} \rightarrow 0,
\end{multline}
thanks to the fact that $\mathrm{tr}(\mathbf{R}_s^2)/M^2 \rightarrow 0$ as $M \rightarrow \infty$. As a consequence of \eqref{eq:Limv1}, under the spatially correlated Rayleigh fading propagation, the convergence of \eqref{eq:FPkkprimeClosed} is very fast as $M$ increases. For the uncorrelated Rayleigh fading model, the deterministic favorable propagation metric is computed as
\begin{equation} \label{eq:FPFactorUn}
\mathsf{FP}_{kl} = \frac{\beta_k \beta_{l} + N \beta_k \xi_{l} + N \beta_{l} \xi_{k} +  N^2 \xi_k \xi_{l} }{M \beta_k \beta_{l} + MN \beta_k \xi_{l} + MN \beta_{l} \xi_{k}  + M N^2 \xi_{k} \xi_{l} } \rightarrow 0,
\end{equation}
as $M, N \rightarrow \infty$ with a fixed ratio. In \eqref{eq:FPFactorUn}, the numerator and denominator respectively scale up in the order of $N^2 \xi_k \xi_{l}$ and $MN^2 \xi_k \xi_{l}$, thus the convergence is proportional to $1/M$. By only considering the spatially correlated model, it explicitly indicates that the deterministic favorable propagation metric converges to zero quickly as more antenna and phase-shift elements are installed. Fig.~\ref{Fig2}(a) plots the deterministic favorable propagation metric in \eqref{eq:FPkkprimeClosed} for a system with the equal number of antennas and phase-shift elements. The covariance matrices are given in \eqref{eq:Cov1} and \eqref{eq:Cov2} with $S_\nu = 3$, $\sigma_\nu = 3^\circ,$ and $d_H = d_V = \lambda/4$. The BS is located at the origin of a Cartesian system coordinate, while RIS is at the $(x,y)-$coordinate $(125, 125)$~m and two users are at $(250,12.5)$~m and $(125, -250)$~m. We observe that the correlated Rayleigh fading channels offer the favorable propagation with the lowest value (the best performance) belonging to the conventional Massive MIMO (denoted as ``direct link" in Fig.~\ref{Fig2}). The worst case  is for a system with only the indirect links (denoted as ``indirect link" in Fig.~\ref{Fig2}). The aggregated channels (denoted as ``both link" in Fig.~\ref{Fig2}) falls in the between. 

We now compute the deterministic channel hardening metric for the spatially correlated Rayleigh fading channel model \eqref{eq:ChannelModel1} as shown in Lemma~\ref{lemma:Hardening}.
\begin{lemma}\label{lemma:Hardening}
If the propagation channels are spatially correlated Rayleigh fading channels by \eqref{eq:ChannelModel1}, the deterministic channel hardening metric is computed in the closed form as in \eqref{eq:CHv1}.
\begin{figure*}
\begin{equation} \label{eq:CHv1}
\mathsf{CH}_k = \frac{ 2 \mathrm{tr}(\mathbf{R}_s \mathbf{R}_k) \mathrm{tr}(\widetilde{\pmb{\Theta}}_k) +  \big|\mathrm{tr}(\widetilde{\pmb{\Theta}}_k ) \big|^2\mathrm{tr}\big(\mathbf{R}_s^2\big) + \mathrm{tr} \big(\widetilde{\pmb{\Theta}}_k^2 \big) \big( M^2 \beta_s^2 + \mathrm{tr}\big(\mathbf{R}_s^2\big)\big) + \mathrm{tr}(\mathbf{R}_k^2)}{\big(M \beta_k + M \beta_s \mathrm{tr}(\widetilde{\pmb{\Theta}}_k) \big)^2}
\end{equation}
\vspace*{-0.1cm}
\hrule
\vspace*{-0.2cm}
\end{figure*}
\end{lemma}
\begin{proof}
Notice that $\mathsf{Var}\{ \| \mathbf{z}_k \|^2 \} = \mathbb{E}\{ \| \mathbf{z}_k\|^4 \} - (\mathbb{E}\{ \| \mathbf{z}_k \|^2 \})^2$, then plugging \eqref{eq:zk2} and \eqref{eq:zk4v1} into \eqref{eq:CHk}, we obtain the result as shown in the lemma.
\end{proof}
The numerator and denominator of \eqref{eq:CHv1} scale up with the order of $N^2M + M^2N$ and $M^2N^2$, thus the deterministic channel hardening metric converges to zero as the number of antennas and phase shifts grows. For the uncorrelated Rayleigh fading channels, one can simplify \eqref{eq:CHv1} to as
\begin{equation} \label{eq:CHUncorr}
\mathsf{CH}_k = \frac{ 2N\beta_k\xi_{k} + N^2\xi_k^2 + (M+1)N \xi_k^2 + \beta_k^2}{M (\beta_k + N\xi_k)^2} \rightarrow 0,
\end{equation}
as $M, N \rightarrow \infty$ with a fixed ratio and the convergence is proportional to $(M+N)/(MN)$. The metric \eqref{eq:CHUncorr} depends on the number of phase-shift elements and antennas that could be optimized to have better channel hardening effects. However, the closed-form expression \eqref{eq:CHUncorr} is independent of the phase-shift coefficients and we do not have a smart environment control based on the channel statistic with the uncorrelated Rayleigh fading channels. Fig.~\ref{Fig2}(b) displays the deterministic channel hardening metric in \eqref{eq:CHv1}, which demonstrates that the channels is less hardened when including the extra channel coefficients from the indirect link under the presence of the RIS. Furthermore, Fig.~\ref{Fig2}(c) shows a higher rank of the covariance matrix $\mathbb{E}\{ \mathbf{h}_k \mathbf{h}_k^H \}$ thanks to the assistance of the RIS with $100$ antennas equipped at the BS. 

\vspace*{-0.2cm}
\section*{What We Have Learned}
%\vspace*{-0.1cm}
\par This lecture note described three fundamental properties of RIS-assisted Massive MIMO channels comprising the  \textit{favorable propagation}, \textit{channel hardening}, and \textit{rank deficiency}. The definitions of those properties were presented along with several measures that are deterministic metrics expressed by channel statistics only. The closed-form expressions of the deterministic metrics were computed for a system with arbitrary (but finite) numbers of antennas and phase shifts.  Numerical results demonstrate that by exploiting the Rayleigh fading model, the aggregated channels still offer the favorable propagation and channel sparsity, but are worse than conventional Massive MIMO systems. We further provided insightful observations on the rank deficiency that occurs when there are a limited number of scatterers in the propagation environments. We approved that this issue can be handled by the assistance of an RIS. 
\vspace*{-0.2cm}
\section*{Acknowledgments}
The  work of T. V. Chien, S. Chatzinotas,  and B. Ottersten  was  supported  by  RISOTTI-Reconfigurable Intelligent Surface for Smart Cities  under  Project FNR/C20/IS/14773976/RISOTTI. The work of H. Q. Ngo was supported by the UK Research and Innovation Future Leaders Fellowships under Grant MR/S017666/1.
\vspace*{-0.2cm}
\section*{Authors}
\textit{Trinh Van Chien} (trinhchien.dt3@gmail.com) received his Ph.D. degree in Electrical Engineering with specialization in communication systems from Link\"oping University, Sweden, in 2020. He is currently a research associate at University of Luxembourg, Luxembourg. His research interests focus on convex optimization problems, theoretical analysis, and machine learning applications for wireless communications and image \& video processing. He is an IEEE member.

\textit{Hien Quoc Ngo} (hien.ngo@qub.ac.uk) is currently a Reader with Queen's University Belfast, U.K, and a UKRI Future Leaders Fellow. He has coauthored many research articles in wireless communications and coauthored the Cambridge University Press textbook ``Fundamentals of Massive MIMO" (2016). His main research interests include Massive MIMO, cell-free Massive MIMO, physical layer security, and cooperative communications. He received the IEEE ComSoc Stephen O. Rice Prize in 2015, the IEEE ComSoc Leonard G. Abraham Prize in 2017, and the Best Ph.D. Award from EURASIP in 2018.  He serves as the Editor for the IEEE Transactions on Wireless Communications, IEEE Wireless Communications Letters.
 
\textit{Symeon Chatzinotas} (symeon.chatzinotas@uni.lu) is currently Full Professor / Chief Scientist I and Head of the SIGCOM Research Group at SnT, University of Luxembourg. He is coordinating the research activities on communications and networking, acting as a PI for more than 20 projects and main representative for 3GPP, ETSI, DVB. He was the co-recipient of the 2014 IEEE Distinguished Contributions to Satellite Communications Award and Best Paper Awards at EURASIP JWCN, CROWNCOM, ICSSC. He is currently in the editorial board of the IEEE Transactions on Communications, IEEE Open Journal of Vehicular Technology and the International Journal of Satellite Communications and Networking.

\textit{Bj\"orn Ottersten} (bjorn.ottersten@uni.lu) (S'87 – M'89 – SM'99 – F'04) received the Ph.D. degree in electrical engineering from Stanford University, Stanford, CA, USA, in 1990. In 1991, he was appointed Professor of signal processing with the Royal Institute of Technology (KTH), Stockholm, Sweden. He is currently the Director for the Interdisciplinary Centre for Security, Reliability and Trust, University of Luxembourg. He is a recipient of the IEEE Signal Processing Society Technical Achievement Award, the EURASIP Group Technical Achievement Award, and the European Research Council advanced research grant twice. 

\bibliographystyle{IEEEtran}
\bibliography{IEEEabrv,refs}
\end{document}